\documentclass[journal]{IEEEtran}
\usepackage[final]{changes}
\usepackage{amsmath,amsfonts}
\DeclareMathOperator*{\argmax}{arg\,max}

\usepackage{algpseudocode}
\usepackage{algorithm}
\usepackage[caption=false,font=normalsize,labelfont=sf,textfont=sf]{subfig}
\usepackage{textcomp}
\usepackage{stfloats}
\usepackage{url}
\usepackage{xurl}
\usepackage{verbatim}
\usepackage{xcolor}
\usepackage{graphicx}
\usepackage{amssymb,amsthm,epsfig,epstopdf,array}
\usepackage{float}
\usepackage{booktabs}
\usepackage{enumitem}
\usepackage{lscape}
\usepackage{siunitx}
\usepackage{pifont}
\usepackage{multirow}
\usepackage{tikz}
\usetikzlibrary{arrows}
\newtheorem{assumption}{Assumption}
\newtheorem{proposition}{Proposition}
\usepackage[numbers]{natbib}

\begin{document}
\title{Optimal Drive-by Sensing in Urban Road Networks with Large-scale Ridesourcing Vehicles}
\author{\IEEEauthorblockN{Shuocheng Guo and Xinwu Qian}
\thanks{Shuocheng Guo and Xinwu Qian \textit{(corresponding author)} are with the Department of Civil, Construction and Environmental Engineering, The University of Alabama, Tuscaloosa, AL 35487, USA (e-mail: sguo18@ua.edu; xinwu.qian@ua.edu)}}
\maketitle

\begin{abstract}
The sensing and monitoring of the urban road network contribute to the efficient operation of the urban transportation system and the functionality of urban systems. 
However, traditional sensing methods, such as inductive loop sensors, roadside cameras, and crowdsourcing data from massive urban travelers (e.g., Google Maps), are often hindered by high costs, limited coverage, and low reliability. 
This study explores the potential of drive-by sensing, an innovative approach that employs large-scale ridesourcing vehicles (RVs) for urban road network monitoring. We first evaluate RV sensing performance by coverage and reliability through historical road segment visits. 
Next, we propose an optimal trip-based RV rerouting model to maximize the sensing coverage and reliability while preserving the same level of service for the RVs' mobility service. Furthermore, a scalable column generation-based heuristic is designed to guide the cruising trajectory of RVs, assuming trip independence.
The effectiveness of the proposed model is validated through experiments and sensitivity analyses using real-world RV trajectory data of over 20,000 vehicles in New York City. 
The optimized rerouting strategy has yielded significantly improved results, elevating explicit sensing coverage of the road network by 15.0\% to 17.3\% (varies by time of day) and achieving an impressive enhancement in sensing reliability by at least 24.6\% compared to historical records. 
Expanding the path-searching space further improved sensing coverage of up to 4.5\% and reliability of over 4.2\%. 
Moreover, considering incentives for RV drivers, the enhanced sensing performance comes at a remarkably low cost of \$0.10 per RV driver, highlighting its cost-effectiveness.
\end{abstract}


\section*{Nomenclature}
\addcontentsline{toc}{section}{Nomenclature}
\subsection{Abbreviations}
\begin{IEEEdescription}[\IEEEusemathlabelsep\IEEEsetlabelwidth{$V_1,V_2$}]
\item[ECR] Explicit Coverage Rate
\item[GPS] Global Positioning System
\item[ImpCR] Implicit Coverage Rate
\item[InfCR] Inferred Coverage Rate
\item[NYC] New York City
\item[OD] Origin-Destination
\item[OSRM] Open Source Routing Machine
\item[RFID] Radio Frequency Identification
\item[RV] Ridesourcing Vehicles
\item[RVRP] RV Rerouting Problem
\end{IEEEdescription}

\subsection{Sets and vectors}

\begin{IEEEdescription}[\IEEEusemathlabelsep\IEEEsetlabelwidth{$V_1,V_2$}]
\item[$\mathcal{G}$] A transportation network, $\mathcal{G}:=\left(\mathcal{N},\mathcal{A}\right)$
\item[$\mathcal{N}$] Set of nodes in $\mathcal{G}$
\item[$\mathcal{A}$] Set of links in $\mathcal{G}$
\item[$\mathcal{A}_{\text{ob}}$] Set of links that are observed
\item[$\mathcal{A}_{\text{inf}}$] Set of links that are not observed but inferable
\item[$\mathcal{A}_{\text{add}}$] Set of minimal additional links to be visited to achieve full coverage, i.e., $\mathcal{A}=(\mathcal{A}_{\text{ob}}+\mathcal{A}_{\text{add}})\cup\mathcal{A}_{\text{inf}}$
\item[$\mathcal{T}$] Set of time slots for the sampled horizon
\item[$\mathcal{M}^{t}$] Set of historical cruising trips at $t\in\mathcal{T}$
\item[$\mathcal{R}^{t},\mathcal{R}^{t,*}$] Set of feasible/optimal alternative trips at $t\in\mathcal{T}$
\end{IEEEdescription}

\subsection{Parameter}
\begin{IEEEdescription}[\IEEEusemathlabelsep\IEEEsetlabelwidth{$V_1,V_2$}]
\item[$t_{ij}$] Travel time on arc $(i,j)\in\mathcal{A}$
\item[$t_{r}$] Total travel time on trip $r\in\mathcal{R}^{t}$
\item[$T_{i}^{o},T_{i}^{d}$] Historical OD pick-up and drop-off time for trip $i\in\mathcal{M}^{t}$
\item[$L_{i}^{o},L_{i}^{d}$] Historical OD locations for trip $i\in\mathcal{M}^{t}$
\item[$K^{t}$] RV fleet size at time $t$
\item[$\alpha_{k}^{r}$] Binary indicator that maps feasible trip $r\in\mathcal{R}^{t}$ to historical trip $k\in\mathcal{M}^{t}$
\item[$\beta_{ij}^{r},\beta_{ij}^{k}$] Binary indicator that maps traversed arc $(i,j)\in\mathcal{A}$ to feasible trip $r\in\mathcal{R}^{t}$ or historical trip $k\in\mathcal{M}^{t}$
\item[$\theta_{D},\theta_{T}$] Unit costs for driving distance and time 
\item[$\eta$] A scaling parameter for incentive determination
\item[$b_{k}^{t}$] Incentive for altering historical trip $k\in\mathcal{M}^{t}$
\end{IEEEdescription}


\subsection{Variables}
\begin{IEEEdescription}[\IEEEusemathlabelsep\IEEEsetlabelwidth{$V_1,V_2$}]
\item[$P_{ij}^{t}$] Sensing frequency of arc $(i,j) \in \mathcal{A}$ at $t\in\mathcal{T}$
\item[$y_{r}$] A binary decision variable, where $y_{r}=1$ indicates feasible route $r$ is included in the optimal solution
\item[$\mathbf{v}^{t},\overline{\mathbf{v}}^{t}$] Visit counts under optimal rerouting and expected visit counts at $t\in\mathcal{T}$
\end{IEEEdescription}

\subsection{Evaluation Metrics}
\begin{IEEEdescription}[\IEEEusemathlabelsep\IEEEsetlabelwidth{$r^{\text{exp}},r^{\text{inf}},r^{\text{imp}}$}]
\item[$S^{t}$] Sensing power at $t\in\mathcal{T}$
\item[$H^{t}$] Entropy of sensing frequency on $\mathcal{A}$ at $t\in\mathcal{T}$
\item[$I^{t}$] Independence score of optimal routes $\mathcal{R}^{t,*}$
\item[$B^{t}$] Incentive per RV driver to reroute at $t\in\mathcal{T}$
\item[$r^{\text{exp}},r^{\text{inf}},r^{\text{imp}}$] Explicit/Inferred/Implicit coverage rates
\end{IEEEdescription}

\section{Introduction}\label{sec:intro}
The \textit{drive-by sensing}, facilitated by the advances in mobile sensors, cloud computing, and the growth of shared mobility services, offers a low-cost and reliable alternative to conventional fixed sensor approaches for monitoring dynamic states in urban road networks, i.e., traffic flow, road congestion, and air quality \cite{yoon2007surface,work2008ensemble}. Conventional sensing methods using fixed sensors, e.g., inductive loop detectors and roadside cameras, are limited by reliability issues and financial concerns related to installation, operation, and maintenance. In contrast, as depicted in Fig.~\ref{fig:overview_drive_by_sensing}, the drive-by sensing approach employs a fleet of mobile sensor-equipped Ridesourcing Vehicles (RVs) operating as 24/7 mobility service providers, collecting crowd-sourced data while traversing the urban road network \cite{mora2019towards,o2019quantifying}. Third-party city services utilize this collected data to enhance traffic and urban environment monitoring. 
The versatility of mobile sensors enables RV-based drive-by sensing to overcome the limitations of traditional sensing methods, which provides a more cost-effective and flexible deployment at scale. As a result, it finds application in various areas, including transportation planning (real-time travel speed prediction \cite{cui2019traffic} and short-term travel demand forecasting \cite{qian2020short}), environmental and infrastructure monitoring \cite{desouza2020air,khan2016integration}, information broadcasting in the connected vehicle environment \cite{abdelhamid2015vehicle}, and training data collection for autonomous driving~\cite{anjomshoaa2018city}. For a more comprehensive review of drive-by sensing applications, see~\citet{anjomshoaa2018city} and \citet{ji2023survey}. Table~\ref{tab:RV_sensing_applications} summarizes the various applications of RV-based drive-by sensing.

\begin{figure}[h]
    \centering
    \includegraphics[width=0.5\textwidth]{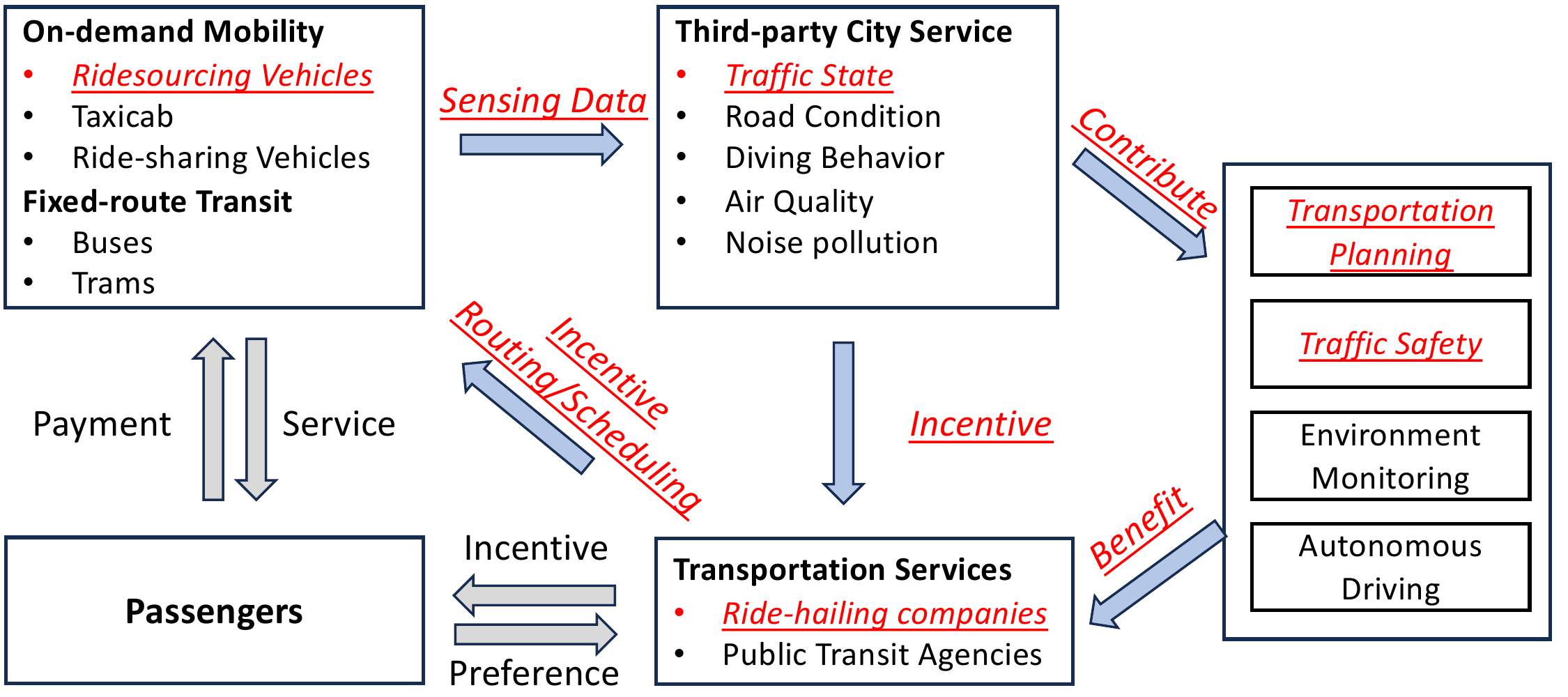}
    \caption{An overview of drive-by sensing (our focus is marked in red)}
    \label{fig:overview_drive_by_sensing}
\end{figure}

\begin{table}[h]
    \centering
    \caption{Possible applications of RV-based driver-by sensing}\label{tab:RV_sensing_applications}
    \begin{tabular}{p{1.7cm}|p{1.3cm}|p{4.4cm}}
        \hline
        Category & Sensor & Application(s)  \\
        \hline
        Ambient Fluid   & Air quality monitor      & Monitoring air quality \\\hline
        Urban Envelope	&  Camera & Real-time imaging of urban areas and collecting street information, e.g., Google Street View~\cite{google2023street}\\\hline
        Electromagnetic &   GPS and RFID    & Monitoring traffic flow and congestion~\cite{zhan2020link}\\\hline
        Electromagnetic &   GPS and gyroscope    & Fetching sample trajectories for training autonomous vehicles\\\hline
        Electromagnetic &	WiFi and Bluetooth   & Constructing vehicle-to-vehicle network \\\hline
        \multicolumn{3}{l}{\footnotesize Note: Adapted from \citet{anjomshoaa2018city}.}
    \end{tabular}
\end{table}

This study addresses three main challenges in executing efficient RV-based drive-by sensing. First, RV trips are driven by passenger demand, leading to disproportional sensing coverage and spatial disparities in the urban road network, with high-demand areas being repeatedly visited~\cite{qian2015characterizing}. Second, RVs, designed primarily for passenger transport, may not deviate from their routes for better sensing coverage when carrying passengers. Only unoccupied RVs can potentially be rerouted for improved sensing performance. However, this is constrained by the RV drivers' willingness to detour and user inconvenience (e.g., timely arrival at passengers). Finally, existing studies indicate that complete system information can be inferred based on partial observations without visiting every road segment~\cite{ng2012synergistic,he2013graphical}. However, the inferred information will depend on explicit visits and road network topology, suggesting the possibility of enhancing sensing performance by selectively visiting a subset of road segments.

To address these challenges, this study proposes an optimal RV-based drive-by sensing strategy by rerouting the \textit{unoccupied} RVs to monitor the large-scale urban road network. The investigation begins by analyzing historical sensing performance, including the proportion of road segments that are explicitly visited, the inferable road segments that are not visited, and the \textit{sensing power}~\cite{o2019quantifying}, to capture the sensing reliability through repeated visits. 
To enhance sensing performance, an optimal RV rerouting model is developed, aiming at simultaneously maximizing explicit sensing coverage and reinforcing sensing reliability. One path-based algorithm is designed to guide unoccupied RVs to road segments with lower sensing frequency, providing high-quality solutions in real time. The effectiveness of the model is demonstrated using real-world RV trajectory data from New York City, covering over 100,000 RVs and more than 1 million daily passenger trips~\cite{nyc_yearbook2020}. It is worth noting that this study is built upon our previous work~\cite{guo2022sensing}, which explored the universality of sensing frequency and RV driving range and serves as a baseline for the optimal rerouting strategy in this study. 
The main contributions of this study can be summarized as follows:

\begin{itemize}
    \item We utilize real-world RV trajectory data to analyze the sensing performance in a large-scale urban road network, focusing on coverage and reliability.
    \item We develop and validate a path-based rerouting model that guides unoccupied RVs to enhance sensing performance in terms of both coverage and reliability.
    \item We propose a scalable column generation-based heuristics based on a decomposition method, facilitating real-time rerouting decisions for scenarios with idling RVs to be assigned.
    \item We demonstrate the effectiveness of the optimal sensing strategy and investigate the trade-offs between sensing performance and total incentive costs for rerouting.
\end{itemize}

The remainder of this study is organized as follows. In the next section, we provide a brief review of related literature. Section~\ref{sec:model} presents the RV rerouting models, solved by a path-based approach in Section~\ref{sec:solution}. The numerical experiment is conducted in Section~\ref{sec:experiment}, where we introduce the dataset, showcase the historical sensing performance, and demonstrate the effectiveness of the optimal rerouting strategy. Finally, we conclude our study in Section~\ref{sec:conclusions}.

 \section{Literature Review}\label{sec:lit}
The sensor placement problem regarding the fixed-location sensors (i.e., inductive loop detectors and roadside cameras) has been extensively investigated in recent decades (for a comprehensive survey, see~\citet{gentili2018review}). Among the first, \citet{lam1990accuracy} proposed two key components to select the optimal sensor locations: traffic flow volume and origin-destination (OD) coverage. \citet{yang1991analysis} evaluated the performance of the OD flow estimation based on the traffic volume counts. Later, researchers proposed different optimization-based frameworks to find the optimal sensor placement considering the travel time estimation and traffic OD flow count~\cite{yang1998optimal,ban2009optimal,zhou2010information,li2011reliable}. Besides, the reliability of sensor location on the travel time estimation~\cite{fujito2006effect,ban2007performance} and OD coverage~\cite{fei2007sensor,fei2013vehicular} has also been investigated. With a limited budget for fixed-location sensors, complete coverage on every road segment can hardly be achieved. In addition to the optimal sensor location design, it motivates the researchers on the link observability problem, where the information (e.g., traffic volumes) can be inferred based on the historical observation and road network topology with no sensors being physically installed~\cite{hu2009identification,ng2012synergistic,viti2014assessing,xu2016robust}. The advances in GPS technology make it feasible and practical to conduct dynamic probing for traffic conditions. In this case, the data are collected from the GPS-enabled
mobile devices or smartphones~\cite{work2008ensemble,herrera2010evaluation}, mobile traffic sensors~\cite{sun2013vehicle}, and Bluetooth sensors~\cite{park2015optimal}.

Recent advancements in mobile portable sensors and efficient online cloud-computing platforms have made \textit{drive-by sensing} feasible and practical~\cite{mora2019towards}, which has been examined in various cities in the world~\cite{o2019quantifying,anjomshoaa2020quantifying,ma2021high,dai2023exploring}. 
Drive-by sensing, particularly using RV fleets, has shown great potential in overcoming data reliability and coverage issues associated with fixed-location sensors, which typically require substantial investment. One distinctive feature of RV-based drive-by sensing is its flexibility, as a slight deviation from original routes can significantly enhance sensing coverage without the need for additional fixed-location sensor construction. In addition, RV-based drive-by sensing fleets offer a robust platform for gathering high-resolution spatial and temporal data \cite{anjomshoaa2020quantifying}, effectively capturing dynamic environmental attributes like air pollutants, air pressure, and spatial temperature variations \cite{anjomshoaa2018city}. However, RV operators might exhibit reluctance in undertaking sensing tasks entailing elevated driving distance to the target area. To address this challenge, various incentive mechanisms have been proposed, which aim to incentivize participants to engage in tasks within less popular areas, either for the fulfillment of crowdsensing assignments \cite{tian2017movement,liu2020incentive} or to stimulate participation rates \cite{zhang2019incentive,wang2020promoting}.  Interested readers are referred to \citet{zhang2015incentives} for a comprehensive survey.

Despite these advantages, few studies have investigated the feasibility and efficiency of large-scale RV-based drive-by sensing in urban areas. Two existing studies focused on measuring the potential sensing performance based on historical taxi trips. \citet{agarwal2020modulo} examined the explicit sensing coverage of drive-by sensing in San Francisco and Rome of 1,138 and 304 vehicles, respectively. \citet{o2019quantifying} proposed sensing power to statistically quantify sensing reliability under different fleet sizes using taxi trip data from nine major cities. Both studies were limited to explicit observations based on historical trips without redeployment to improve sensing performance. 
In contrast, our study proposes an optimal RV rerouting strategy to explore the potential improvement in the sensing performance. Furthermore, we investigate the supplementary incentive costs required to improve the sensing performance.

\section{Model and formulation}\label{sec:model}
In this section, we propose an optimal RV rerouting strategy with the objective of enhancing the sensing coverage and reliability by altering the cruising paths for each unoccupied RV in the historical data. We start by defining the problem and discussing the properties of the objective function. Next, we present the trip-based formulation of the RV rerouting strategy. Finally, we emphasize the challenges and distinctive characteristics of the RV rerouting model.

\subsection{Problem Definition}
We focus on the RV Rerouting Problem (RVRP), which is extended from the vehicle rerouting/rescheduling problem with time windows~\cite{li2007vehicle,li2009real}. The RVRP is defined on a transportation network $\mathcal{G}=(\mathcal{N},\mathcal{A})$, where the node set $\mathcal{N}$ includes road intersections and link set $\mathcal{A}$ represents the road segments. 
Let $\mathcal{M}^{t}$ denote the set of historical cruising trips at time interval $t\in\mathcal{T}$. For each historical trip $k\in\mathcal{M}^{t}$, we define a binary indicator $\beta_{ij}^{k}$ that maps the trip to the link level, where $\beta_{ij}^{k}=1$ indicates that the arc $(i,j)\in\mathcal{A}$ is traversed by historical trip $k$, and $0$ otherwise. Additionally, for each trip $k\in\mathcal{M}^{t}$, the OD locations are represented by $L_{k}^{o}$ and $L_{k}^{d}$, and the corresponding pick-up and drop-off time schedules are $T_{k}^{o}$ and $T_{k}^{d}$, respectively. 
Our objective is to find a subset of alternative cruising trips, $\mathcal{R}^{t}$, between two known locations $L_{k}^{o}$ and $L_{k}^{d}$, for each historical trip $i\in\mathcal{M}^{t}$, to achieve improved sensing coverage and reliability, which are quantified as follows.

\paragraph{Sensing coverage} The sensing coverage rate is measured as the proportion of explicitly visited road segments, referred to as the explicit coverage rate (ECR), denoted by $r^{\text{exp}}$. The ECR is expressed as follows:
\begin{equation}
r^{\text{exp}} = \frac{\sum_{(i,j)\in\mathcal{A}} \min\{1,\sum_{k\in\mathcal{M}^{t}}\beta_{ij}^{k}\}}{|\mathcal{A}|}\label{eq:ECR}
\end{equation}
where the denominator $|\mathcal{A}|$ represents the total number of road segments. The numerator counts the number of traversed road segments. If any of the historical trips $k\in\mathcal{M}^{t}$ traverses arc $(i,j)\in\mathcal{A}$, we consider the arc $(i,j)$ as a visited road segment. Conversely, if none of the trips $k\in\mathcal{M}^{t}$ cover the arc $(i,j)$, we consider the arc $(i,j)$ as an unvisited road segment.

\paragraph{Sensing reliability}
We quantify the sensing reliability by introducing the concept of \textit{sensing power}~\cite{o2019quantifying}. The sensing power at time $t$, denoted by $S^{t}$, helps to examine if the RV fleet can reproduce a similar sensing pattern during different time periods $t\in\mathcal{T}$ given the randomness of the driver's behavior and the distribution of travel demand. 
Formally, the sensing power denotes the expectation on the probability that one road segment $(i,j) \in \mathcal{A}$ can be traversed at least once by the RV fleet during the time period $t\in\mathcal{T}$. 
We express the sensing power $S^{t}$ as follows:

\begin{equation}
    S^{t} \approx 1 - \frac{1}{|\mathcal{A}|} \sum_{(i,j) \in \mathcal{A}}\left(1-P_{ij}^{t}\right)^{N^{t}}
    \label{eq:sensing_power}
\end{equation}
where $P_{ij}^{t}$ is the sensing frequency of road segment $(i,j)\in\mathcal{A}$ at time $t$ and $\sum_{(i,j)\in\mathcal{A}}P_{ij}^{t}=1$ for all $t\in\mathcal{T}$. The power term $N^{t}$ denotes the total number of traversed road segments by all RV historical trips $\mathcal{M}^{t}$, which can be calculated as $N^{t}=\sum_{k\in\mathcal{M}^{t}}\sum_{(i,j)\in\mathcal{A}}\beta_{ij}^{k}$. As a result,
the term $\left(1-P_{ij}^{t}\right)^{N^{t}}$ represents the probability of road segment $(i,j)\in\mathcal{A}$ not being traversed by the RV trips $\mathcal{M}^{t}$ during period $t$.

\subsection{Trip-based Formulation}
In this section, we introduce the trip-based formulation of the RVRP. Let $\mathcal{R}^{t}$ be the set of alternative cruising trips at time period $t\in\mathcal{T}$. Each trip $r$ in the set $\mathcal{R}^{t}$ can be represented as a sequence of nodes $r:=(r_0,r_1,\ldots,r_{n-1},r_{n})$, where the corresponding arrival time schedule at each traversed node is $T_{r}^{0},T_{r}^{1},\ldots,T_{r}^{n-1},T_{r}^{n}$. Let $t_{r}$ denote the total travel time on alternative trip $r\in\mathcal{R}^{t}$, an altered trip $r$, corresponding to historical trip $k\in\mathcal{M}^{t}$, is considered feasible if the following conditions are satisfied.
\begin{itemize}
    \item $r_{0}=L_{k}^{o}$ and $r_{n}=L_{k}^{d}$
    \item $T_{r}^{0}\geq T_{k}^{o}$ and $t_{r}\leq T_{k}^{d}-T_{k}^{o}$
\end{itemize}

To achieve the objectives of maximizing both sensing coverage and reliability, we formulate our objective function as the information entropy of the sensing frequency $P_{ij}^{t}$ across all road segments $(i,j)\in\mathcal{A}$, denoted as $H^{t}$. A low entropy value indicates a skewed distribution of vehicle visiting frequency on each road segment, while a high entropy value signifies a more evenly distributed visiting count. Consequently, the RV rerouting strategy aims to enhance sensing coverage and reliability by selecting alternative paths $r\in\mathcal{R}^{t}$ for each historical trip $k\in\mathcal{M}^{t}$ that covers more road segments with low sensing frequency (including unvisited segments) and fewer road segments with high sensing frequency (e.g., repetitive visits). Parameter $\alpha_{k}^{r}$ is a binary indicator with $\alpha_{k}^{r}=1$ representing that historical trip $k\in\mathcal{M}^{t}$ is altered by trip $r\in\mathcal{R}^{t}$, and $0$ otherwise. The binary variable $y_{r}$ equals $1$ if trip $r$ is selected to be part of the optimal solution and $0$ otherwise. 
The trip-based formulation is expressed as follows:

\begin{subequations}
    \begin{align}
        \max_{\mathbf{y}}{H^{t}}= & \sum_{(i,j)\in\mathcal{A}} - P_{ij}^{t} \log(P_{ij}^{t})\label{eq:obj}\\
        \text{s.t.~} & P_{ij}^{t} = \frac{\sum\limits_{r\in\mathcal{R}^{t}}\beta_{ij}^{r}y_{r}}{\sum\limits_{(i,j)\in\mathcal{A}}\sum\limits_{r\in\mathcal{R}^{t}}\beta_{ij}^{r}y_{r}}\label{eq:cons1}\\
        & \sum\limits_{r\in\mathcal{R}^{t}}\alpha_{k}^{r}y_{r}=1 && \forall k\in\mathcal{M}^{t}\label{eq:cons2}\\
        & y_{r}\in \{0,1\} \label{eq:cons3}
\end{align}
\end{subequations}
where the objective function~\eqref{eq:obj} aims to maximize the information entropy of the sensing frequency. The sensing frequency on arc $(i,j)$, $P_{ij}^{t}$, is defined in Constraint~\eqref{eq:cons1} as the visit counts on the specific arc $(i,j)$ over all arcs $\mathcal{A}$. The numerator counts the sensed road segments over all altered trips $r\in\mathcal{R}^{t}$. The denominator represents the visit counts among all road segments for all trips. 
Constraints~\eqref{eq:cons2} ensure a one-to-one mapping between the historical trip and altered trips, such that exactly one altered trip is selected for each historical trip $k\in\mathcal{M}^{t}$.

We claim in Proposition~\ref{prop:equivalent_condition} that the objective function in Eq.\eqref{eq:obj} is equivalent to maximizing the sensing power in Eq.~\eqref{eq:sensing_power} under the same set of constraints~\eqref{eq:cons1}-\eqref{eq:cons3}. To establish this equivalence, we first demonstrate that both objective functions are concave within the same feasible region, given identical sets of constraints, as we are considering the same subset of RV historical trips $\mathcal{M}^{t}$ in the same transportation network $\mathcal{G}$. Hence, it is sufficient to show that the solution to the objective function in Eq.~\eqref{eq:obj} also satisfies the maximization of the sensing power.

\begin{proposition}
Suppose $\mathcal{R}^{t,\star}$ is the optimal solution to the entropy maximization of $H^{t}$ in Eq.~\eqref{eq:obj}, then $\mathcal{R}^{t,\star}$ can also maximize the sensing power in Eq.~\eqref{eq:sensing_power}. \label{prop:equivalent_condition} 
\end{proposition}
\begin{proof}
According to the inequality of arithmetic and geometric means, we can rewrite the objective function~\eqref{eq:obj} as follows:
\begin{align*}
\exp{H^{t}} = \exp{\left(\sum_{(i,j)\in\mathcal{A}}-P_{ij}^{t}\log(P_{ij}^{t})\right)}& = \prod_{(i,j)\in\mathcal{A}}(P_{ij}^{t})^{P_{ij}^{t}} \\
&\leq \sum_{(i,j)\in\mathcal{A}} P_{ij}^{t}\frac{1}{P_{ij}^{t}} = |\mathcal{A}|
\end{align*}
where the equality holds if and only if $P_{ij}^{t}=\frac{1}{|\mathcal{A}|}$.

On the other hand, for the sensing power, we have:
\begin{align*}
1-\frac{1}{|\mathcal{A}|} \sum_{(i,j)\in \mathcal{A}}\left(1-P_{ij}\right)^{N^{t}}
&\le 1-\left(\prod_{(i,j)\in\mathcal{A}}\left(1-P_{ij}\right)^{N^{t}}\right)^{\frac{1}{|\mathcal{A}|}}~\label{eq:geometry_mean_ineq}\\
&= 1-\left(\prod_{(i,j)\in\mathcal{A}}\left(1-P_{ij}\right)\right)^{\frac{N^{t}}{|\mathcal{A}|}}
\end{align*}
where the equality also holds if and only if $P_{ij}^{t}=\frac{1}{|\mathcal{A}|}$.

Also, both objective functions in Eqs.~\eqref{eq:obj} and \eqref{eq:sensing_power} are concave. 
Hence, the equivalence holds between the information entropy and sensing power.
\end{proof}

\added{Proposition~\ref{prop:equivalent_condition} establishes a connection between entropy maximization and sensing power maximization. In addition, the entropy maximization problem promotes the positive value of $P_{ij}^{t}>0$, which indicates the explicit visits on road segment $(i,j)$, thereby enhancing the overall explicit sensing coverage.}

\section{Solution Approach}\label{sec:solution}
The RVRP~\cite{li2009real} is known to be a challenging problem as it generalizes the Vehicle Routing Problem, which is NP-hard in a strong sense~\cite{solomon1987algorithms}. One approach to solving the RVRP is the Column Generation (CG) technique~\cite{desrochers1992new}, which decomposes the problem into the trip level and iteratively generates feasible trips. Nonetheless, owing to the interdependencies among trips across extended temporal horizons and the exponential surge in trip enumerations inherent to expansive road networks, the exact CG approach becomes computationally burdensome.

Rather than pursuing an exact solution for the RVRP, we propose a CG-based heuristic. This heuristic yields a near-optimal solution considering a collective subset of historical trips over a relatively concise temporal horizon (e.g., 3 hours), which ensures both scalability and computational efficiency for real-world operations.


\subsection{\added{Scalable Heuristic-Driven Column Generation Algorithm}}
This section proposes the Scalable Heuristic-driven Column Generation (SHCG) algorithm. This heuristic approach decomposes the RVRP for each time period $t\in\mathcal{T}$ into independent subproblems for each historical cruising trip $k\in\mathcal{M}^{t}$. In this case, we have $\mathcal{R}^{t}=\bigcup_{k\in\mathcal{M}^{t}}\mathcal{R}_{k}^{t}$, and the objective is to generate a subset of feasible trips $\mathcal{R}_{k}^{t}$ subject to trip $k$.

The SHCG algorithm proceeds by arranging all trips $k\in\mathcal{M}^{t}$ within each time interval $t$ in ascending order of their pick-up time schedule $T_{k}^{o}$. Subsequently, the K-shortest path algorithm~\cite{yen1971finding} with a predefined value of $K$ is applied to find $|\mathcal{R}_{k}^{t}|=K$ feasible paths until the time constraint is violated (e.g., $T_{r}^{0}+t_{r}>T_{k}^{d}$ for a trip $r\in\mathcal{R}_{k}^{t}$). This process ensures the existence of a viable solution for each historical cruising trip, equating to the original historical trip.

Importantly, the SHCG algorithm attains exact solutions for the RVRP when the search range $K$ extends to infinity. However, this advancement is met with a polynomial increase in time complexity (see Section~\ref{sec:time_complexity}). Furthermore, while optimal trips are readily identified within a narrow search range, the enhancement in sensing performance can be marginal. Hence, a strategic selection of the $K$ range is imperative, effectively balancing the trade-off between computational load and improvements in sensing performance.
The detailed solution procedure is summarized in Algorithm~\ref{algo:k-pba}.

\begin{algorithm}[H]
\footnotesize
	\caption{SHCG: Scalable Heuristics-driven CG Algorithm} 
        \textbf{Input:} Historical trip information $\{(L_{k}^{o},L_{k}^{d},T_{k}^{o},T_{k}^{d}):k\in\mathcal{M}^{t},t\in\mathcal{T}\}$; Transportation network $\mathcal{G}=\left(\mathcal{N},\mathcal{A}\right)$; Number of shortest paths $K$. Time window relaxation $\delta$\\
        \textbf{Output:} Set of optimal rerouting trips $\mathcal{R}^{t,*}$\\
	\begin{algorithmic}[1]
	    \State \textbf{Sort} historical trips $M^{t}$ from earliest to latest for sequential decision-making
            \For{$k \in \mathcal{M}^{t}$}\Comment{Obtain set of candidate rerouting trips $\mathcal{R}_{k}^{t}$}
                \State \textbf{Obtain} $\mathcal{R}_{k}^{t}\leftarrow$\textsc{KShortestPath}$\left(\mathcal{G}, K, \delta,\{L_{k}^{o},L_{k}^{d},T_{k}^{o},T_{k}^{d}: k\in\mathcal{M}^{t}\}\right)$
                \State \textbf{Initiate} descent directions $\Delta \mathbf{H}^{t}=\varnothing$
                \For{$r\in \mathcal{R}_{k}^{t}$}\Comment{Check each candidate shortest path}
                    \If{$t_{r}\leq \delta\left(T_{m}^{d}-T_{m}^{o}\right)$} \Comment{Feasibility of time window}
                        \State \textbf{Calculate} sensing frequency $P_{ij}^{t}$ following Eq.~\eqref{eq:cons1} by replacing $\mathcal{R}^{t}$ with $\mathcal{R}^{t,*}$
                        \State \textbf{Calculate} sensing entropy $H^{t}$ following Eq.~\eqref{eq:obj}
                        \State \textbf{Update} the increment on sensing entropy $\Delta \mathbf{H}^{t}\leftarrow \Delta \mathbf{H}^{t}\cup H^{t}$
                    \EndIf
                \EndFor
            \State\textbf{Obtain} optimal trip with descent direction $r_{k}^{*}\leftarrow \argmax_{r} \Delta \mathbf{H}^{t}$ 
            \State \textbf{Update} the set of optimal rerouting trips $\mathcal{R}^{t,*}\leftarrow \mathcal{R}^{t,*}\cup \{r_{k}^{*}\}$
            \EndFor
        \State \textbf{Return} set of optimal rerouting trips $\mathcal{R}^{t,*}$, altering historical trips $\mathcal{M}^{t}$
    \end{algorithmic}
    \label{algo:k-pba}
\end{algorithm}

\subsection{Assumption on trip independence}
The SHCG algorithm decomposes the RVRP to trip level and attains high-quality solutions by sequentially handling the historical trips. It relies on the assumption stated below.

\begin{assumption}[Trip Independence]
    The optimal trips $r\in\mathcal{R}^{t,*}$ (i.e., $y_{r}=1$) is independent in small time intervals.\label{assumption:independence}
\end{assumption}

The assumption of trip independence is based on the observation that the rerouting decisions for each trip $k\in\mathcal{M}^{t}$ in a real-world system may only have weak dependencies on each other, as evidenced by empirical observations in Fig.~\ref{fig:before_after_optimal_c}. In practice, only a few drivers might require simultaneous route decisions at fine time scales (e.g., every few seconds), making it challenging to coordinate routing needs for future cruising trips in small time intervals. Moreover, the OD locations of each cruising trip can vary significantly in a large-scale road network, implying that the rerouting trips may also be spatially independent and solvable separately. Considering these factors, we decompose the original problem into a series of routing decisions for each cruising trip $k\in\mathcal{M}^{t}$. Each cruising trip aims to maximize the current entropy value of sensing frequency myopically, based on decisions made by previous trips, while ensuring feasibility.

Based on Assumption~\ref{assumption:independence}, the solution to the RVRP, achieved by sequentially handling the historical trips, results in an optimal solution. Specifically, it suggests that selecting one altered trip $r\in\mathcal{R}^{t}$ for a historical trip $k\in\mathcal{M}^{t}$ does not influence the selection of altered trips for subsequent historical trips in $\mathcal{M}^{t}$. Therefore, under this assumption, the RVRP can be effectively solved using a CG-based algorithm that sequentially considers each historical trip, allowing the construction of a feasible and optimal set of alternative trips $\mathcal{R}^{t,*}$ for the unoccupied RVs at time $t$.

To ensure the validity of this assumption, we propose the \textit{independence score} ($I^{t}$), which measures the extent to which the optimal trips satisfy the spatial-independence conditions. The independence score $I^{t}$ is calculated based on the cosine distance between the optimal sensing visits ($\mathbf{v}^{t}$) and the expected link flow $\overline{\mathbf{v}}^{t}$, defined as follows:
\begin{subequations}
    \begin{equation}
        \mathbf{v}^{t}=\left(\sum\limits_{r\in\mathcal{R}^{t,*}}\beta_{ij}^{r}: (i,j)\in\mathcal{A}\right)
    \end{equation}
    \begin{equation}
        \overline{\mathbf{v}}^{t}= \left(\min\{\sum\limits_{r\in\mathcal{R}^{t,*}}\beta_{ij}^{r},1\}: (i,j)\in\mathcal{A}\right)
    \end{equation}
Here, each element $v^{t}_{ij}\in \mathbf{v}^{t}$ represents the traversed counts on $(i,j)\in\mathcal{A}$ during time period $t$. The expected link flow is considered as a binary vector with the length of $|\mathcal{A}|$, indicating the case that all visited road segments are exactly traversed once, and unvisited segments are zero.

The independence score $I^{t}$ is then calculated as the cosine similarity between $\mathbf{v}^{t}$ and $\overline{\mathbf{v}}^{t}$, ranging from 0 to 1:
\begin{equation}
    I^{t}=\frac{\mathbf{v}^{t}\cdot \overline{\mathbf{v}}^{t}}{||\mathbf{v}^{t}||||\overline{\mathbf{v}}^{t}||}
\end{equation}
where independence score $I^{t}=1$ indicates that the two vectors are perfectly same. It implies that the trips $r$ in optimal rerouting strategy are independent and each link $(i,j)\in\mathcal{A}^{t,*}$ is traversed exactly once. Conversely, $I^{t}=0$ represents two completely different vectors (e.g., all links are visited more than once in our case).
\end{subequations}

\subsection{Time complexity of SHCG Algorithm}\label{sec:time_complexity}
The time complexity of Algorithm~\ref{algo:k-pba} is $O(|\mathcal{M}^{t}|K(|\mathcal{N}|+\mathcal{A})|\log \mathcal{N})$, which depends on parameters $K$ and the scale of road network $\mathcal{G}=(\mathcal{N},\mathcal{A})$. Specifically, parameter $K$ determines the size of the set of alternative trips $|\mathcal{R}_{k}^{t}|$, where a larger $K$ leads to a more exhaustive search for potential trips. However, increasing $K$ can significantly escalate computational time, posing challenges for large-scale experiments. 
Additionally, the relaxation on time window $\delta$ also has an impact on both solution quality and computational efficiency. Specifically, $\delta$ represents the level of relaxation in the time window, expressed as $\delta(T_{k}^{d}-T_{k}^{o})$, allowing for additional travel time compared to the historical cruising trip. A higher value of $\delta$ results in more alternative relocation routes, potentially leading to better sensing performance. 
An extreme case of the SHCG algorithm is when $K=1$, which is equivalent to identifying the route with the lowest sensing frequency using Dijkstra's Algorithm ($O(|\mathcal{A}|\log |\mathcal{N}|)$). In this scenario, the RVs are sequentially assigned to routes that are rarely visited while still satisfying the time window constraints. By varying the values of $K$ and $\delta$, we aim to identify an appropriate combination of $K$ and $\delta$ that strikes a balance between achieving a high-quality solution and maintaining acceptable computational efficiency.

\section{Numerical Experiment}\label{sec:experiment}

\subsection{Data processing}\label{sec:data}
We use the RV trajectory data in Manhattan, New York City (NYC), to gain insights into the sensing range of the RV fleet over the urban road network. We collected the RV trajectory data in April 2017 using the method as described in~\cite{qian2020impact}. The collected trajectory records include the information of timestamp (in Unix), latitude, longitude, and driver ID (only the last six letters are shown here). A sample of collected trajectory records consists of RV driver's ID, time, and real-time location (latitude and longitude).

We fetched around \SI{100}{GB} of data per day in 2017 by constructing 470 data collection stations. The data quality was validated by inferring the number of complete trips and comparing it with the historical trips reported by the NYC Taxi \& Limousine Commission for every 15-minute time interval~\cite{qian2020understanding}. Specifically, one complete trip is identified if (1) the time gap (\SI{1}{min}$\leq \Delta t\leq$\SI{120}{min}) and spatial displacement ($\Delta d\geq$\SI{400}{m}) between consecutive records or (2) the record was the last trajectory identified for the driver ID. We consider the trip with a time gap $\Delta t>$ \SI{120}{min} as an online/offline transition. In this case, we split the trajectories into two shifts and conduct the trajectory processing within the shifts.

After identifying occupied trips from the trajectory of vacant vehicles, we next augment the initial set of vacant trajectories with the inferred trajectories of the occupied trips. For the trajectories of the vacant trips, we only adopt the coordinates per 15-second interval. As for the occupied trips, only the trip's origin-destination (OD) information is derived, and the detailed trajectory cannot be immediately observed for at least $10$ minutes due to the data limitation. In light of this, we next infer the occupied trip trajectory using the trip's OD information by querying the shortest path in the road network via the Open Source Routing Machine (OSRM)~\cite{luxen-vetter-2011}. The query results from the OSRM consist of the coordinates of traversed intersections and the associated timestamps. Based on that, we further interpolate the sequence of coordinates into 15-second intervals for the sake of consistency. 

We next match the coordinates to the nearest road segments in the \textit{backbone} network in Manhattan in NYC. In this case, we use sequences of road IDs to represent the trajectories of each RV. Note that the backbone road network only includes `primary' and `secondary' roads according to~\cite{boeing2017osmnx,OpenStreetMap}, which are specified below:
3

These backbone roads primarily cover the urban area and are of priority to be sensed. Moreover, the backbone road network can (1) reduce the computational burden under a smaller problem size and (2) diversify the choices of cruising routes in our solution algorithm.

\subsection{Evaluation Metrics}\label{sec:eval}
This section introduces the metrics to evaluate the sensing performance. We also present the incentive mechanism to investigate the supplementary incentives for the improvement of sensing performance. 

\subsubsection{Link observability}
We aim to sense the transportation network $\mathcal{G}=(\mathcal{N},\mathcal{A})$ on the road segment level. Along with the ECR ($r^{\text{exp}}$) defined in Eq.~\eqref{eq:ECR} based on the observable segments $\mathcal{A}_{\text{ob}}$, we extend our focus to include inferrable information based on explicit observations and network topology. To illustrate this, let's consider a transportation corridor comprising $2N$ road segments. By deploying fixed sensors on every other road segment, the entire corridor can be fully observed using only $N$ sensors. This observation leads us to refine our understanding of link observability in two metrics: 
\begin{itemize}
    \item  $\mathcal{A}_{\text{inf}}$: the number of inferable links based on existing observations on links ($\mathcal{A}_{\text{ob}}$).
    \item  $\mathcal{A}_{\text{add}}$: the minimum number of additional links to be sensed to ensure that the entire network can be either directly observed or inferred, i.e., $\mathcal{A}=(\mathcal{A}_{\text{ob}}+\mathcal{A}_{\text{add}})\cup\mathcal{A}_{\text{inf}}$.
\end{itemize}

Based on that, we introduce two new sensing coverage measures: the inferred coverage rate (InfCR), denoted as $r^{\text{inf}}$, and the implicit coverage rate (ImpCR), denoted as $r^{\text{imp}}$. These measures are formally expressed as follows:

\begin{subequations}
\begin{equation}
r^{\text{inf}} = \frac{|\mathcal{A}_{\text{ob}}|+|\mathcal{A}_{\text{inf}}|}{|\mathcal{A}|}
\label{eq:inferred-coverage-rate}
\end{equation}

\begin{equation}
r^{\text{imp}} =\frac{|\mathcal{A}_{\text{ob}}|}{|\mathcal{A}_{\text{ob}}|+|\mathcal{A}_{\text{add}}|}
\label{eq:observation-coverage-rate}
\end{equation}
\end{subequations}

Interested readers can refer to~\citet{ng2012synergistic} and \citet{guo2022sensing} for a more detailed implementation procedure of finding the inferable and additional links ($\mathcal{A}_{\text{inf}}$ and $\mathcal{A}_{\text{add}}$) based on existing observations.

\subsubsection{Incentive for RV rerouting}
We consider incentivizing RV drivers to follow the optimal rerouting strategy. This is achieved through the computation of rerouting costs for each RV driver, based on the revenue paradigm established in the context of order fulfillment~\cite{nyc_driver2016}. 
Specifically, we adopt a linear function that accounts for both the cruising distance $\Delta D_{k}$ and the duration $\Delta T_{k}$ of a given trip $k$. Parameters $\theta_{D}$ and $\theta_{T}$ serve as coefficients for unit distance and time, respectively, which is assumed to be $\theta_{D}=0.631$ and $\theta_{T}=0.287$ based on the empirical evidence~\cite{nyc_driver2016}. Recognizing the various potential outcomes of implementing the rerouting strategy, which may lead to either savings or additional cruising distance and idling time for trips $k\in\mathcal{M}^{t}$, we exclusively consider positive values. These values represent the additional incentive required to accommodate the RV drivers. The rerouting cost for trip $k\in\mathcal{M}^{t}$, denoted as $b^{t}_{k}$, is formulated as follows:
\begin{subequations}
\begin{equation}
    b^{t}_{k} = \theta_{D}\max\{0,\Delta D_{k}\}+\theta_{T}\max\{0,\Delta T_{k}\}
\end{equation}

Furthermore, we propose the incentive mechanism for the sensing task, drawing inspiration from previous works on incentives in movement-based crowdsensing~\cite{tian2017movement,liu2020incentive}. We define an  incentive to alter the historical trip $k$, denoted as $B^{t}_{k}$, expressed as an exponential function:

\begin{equation}
    B_{k}^{t} = e^{\eta \cdot b_{k}^{t}} - 1
\end{equation}

Here, $\eta$ represents a scaling parameter tied to the rerouting cost. In Section~\ref{sec:sen_on_budget}, we will delve into comprehensive sensitivity analyses on $\eta$ to investigate the optimal balance between $\eta$ and the improvement in the sensing performance.

Finally, let $K^{t}$ denote the RV fleet size at time $t$, and the average incentive for each RV driver is determined by:
\begin{equation}
    B^{t}= \frac{1}{K^{t}}\sum_{k\in\mathcal{M}^{t}} B_{k}^{t}
\end{equation}

\end{subequations}
\subsection{Results of historical trips}\label{sec:historical_results}
This section presents the empirical results based on the historical records. To gain a comprehensive understanding of the time-varying sensing performance, we divided the RV trip trajectory into six three-hour time intervals based on the trip starting time, ranging from 9 AM to 3 AM the next day. Note that the trajectory information between 3 AM and 9 AM is unavailable due to data limitations. In our analysis, we focus on the spatial and temporal distribution of sensing frequency, the potential for information inference, and the sensing reliability using historical RV trajectory data. 
To facilitate better descriptions in the following subsections, we introduce the terms \textit{main roads} and \textit{branch roads} to characterize two types of road segments. Specifically, main roads are those with higher connectivity and have a greater potential for information inference. On the other hand, branch roads are secondary roads that may connect to main roads or other branch roads, which are often observed in more distant regions with lower sensing frequency.

\subsubsection{Historical sensing frequency}
This section presents the spatial distribution of sensing frequency in Manhattan, NYC, for three selected time periods on a Wednesday (April 19, 2017). Specifically, Fig.~\ref{fig:veh_sensing_points} shows hot-spot areas near Turtle Bay neighborhoods (circled), where certain links attract more than 500 visits across all three-hour periods. Downtown Manhattan (squared) also serves as another hot spot, especially during 3 PM - 6 PM and 9 PM - 12 AM, indicating extensive short-length trips of unoccupied RVs to and from this area. These hot spots have higher sensing frequency, allowing vehicles to serve available requests promptly and find potential passengers quickly. In total, these two areas account for 29.3\%, 21.5\%, and 24.0\% of all visits during the three time periods.

\begin{figure}[h]
\centering
    \includegraphics[width = 1.0\linewidth]{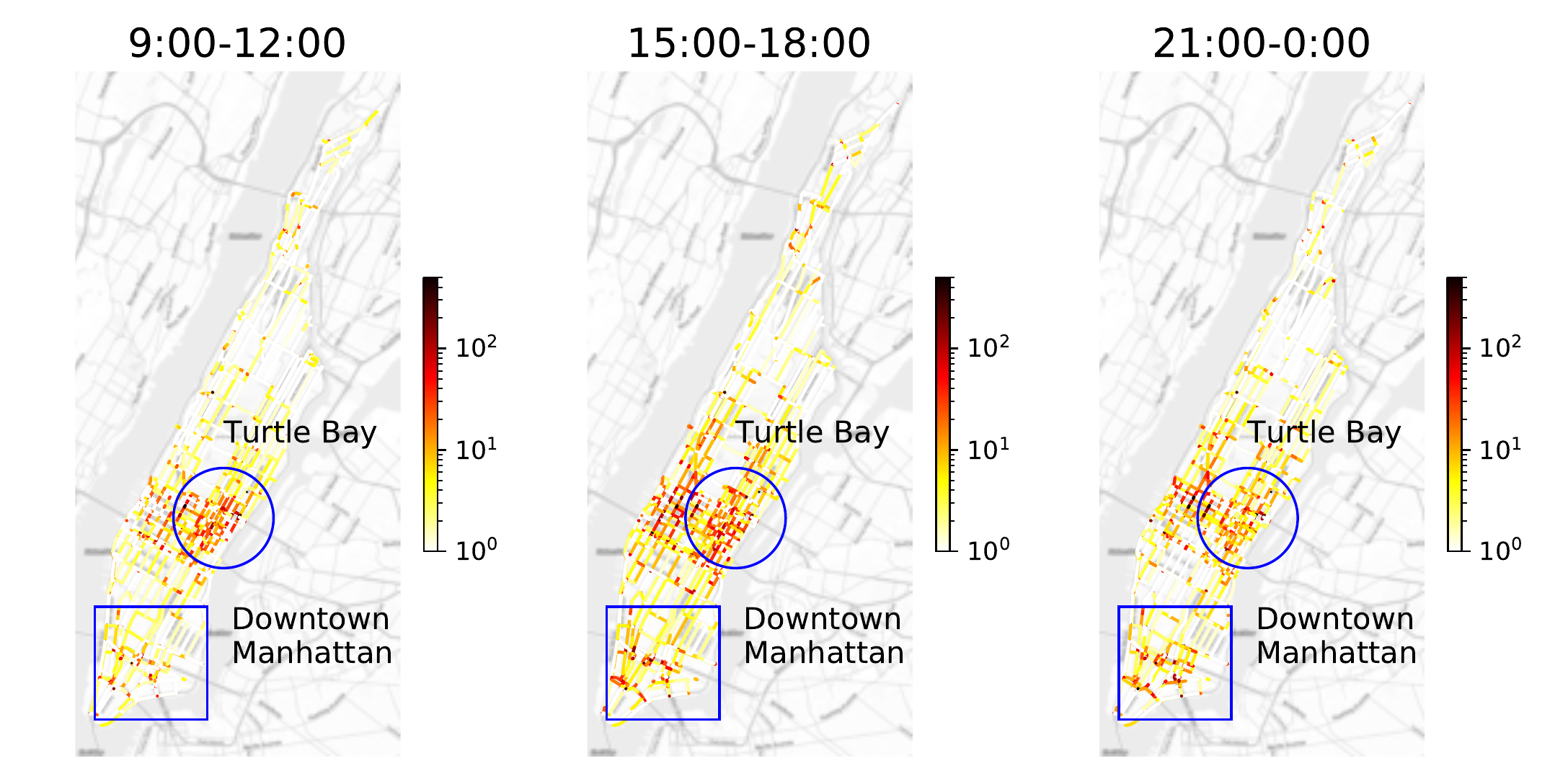}  \caption{Spatial distribution of the RV sensing frequency (on a log scale of 10). The circled area indicates the hot spots near Turtle Bay, and the squared area is Downtown Manhattan.}\label{fig:veh_sensing_points}
\end{figure}

\subsubsection{Historical sensing coverage and reliability}
We next present the weekly average historical sensing performance, focusing on the temporal distribution of coverage rates from 9 AM to 3 AM (of the next day), the sensing power under different fleet sizes, and the independence scores during six different time periods to validate the assumption of trip independence (refer to Assumption~\ref{assumption:independence}).

\begin{figure}[H]
    \subfloat[Sensing Coverage: ECR, InfCR, and ImpCR]{\includegraphics[width=0.5\textwidth]{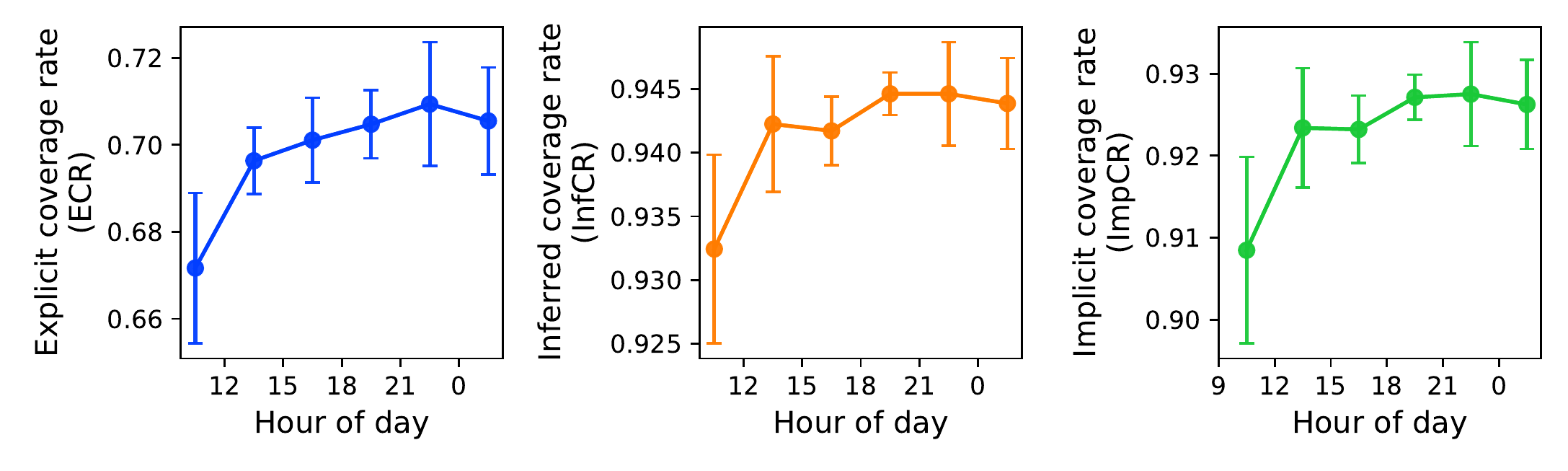}\label{fig:info_infer}}
    \hfill
    \subfloat[Sensing power under different fleet sizes]{\includegraphics[width = 0.17\textwidth]{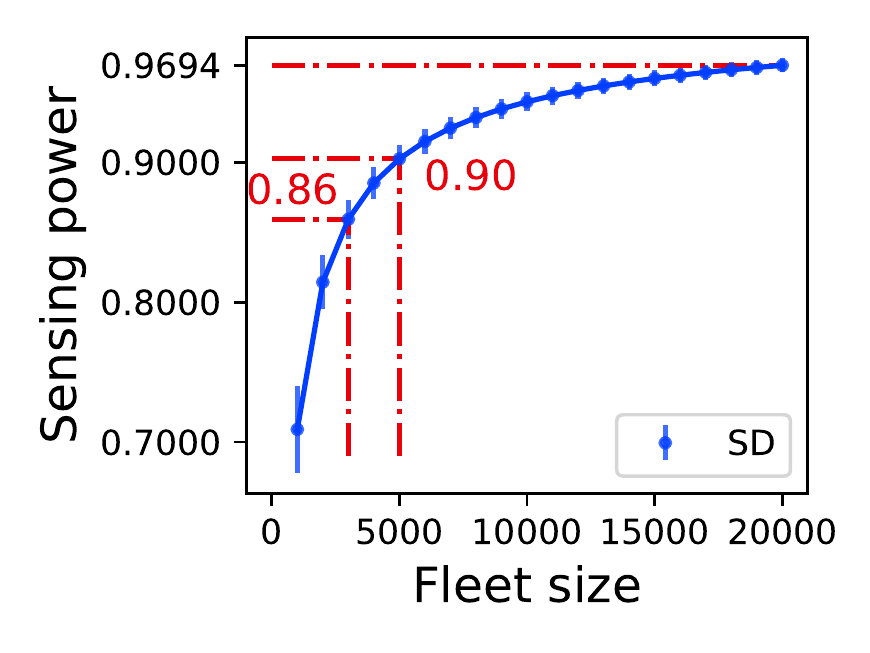}\label{fig:sensing_power}}
    \subfloat[ECR under different fleet sizes]{\includegraphics[width = 0.17\textwidth]{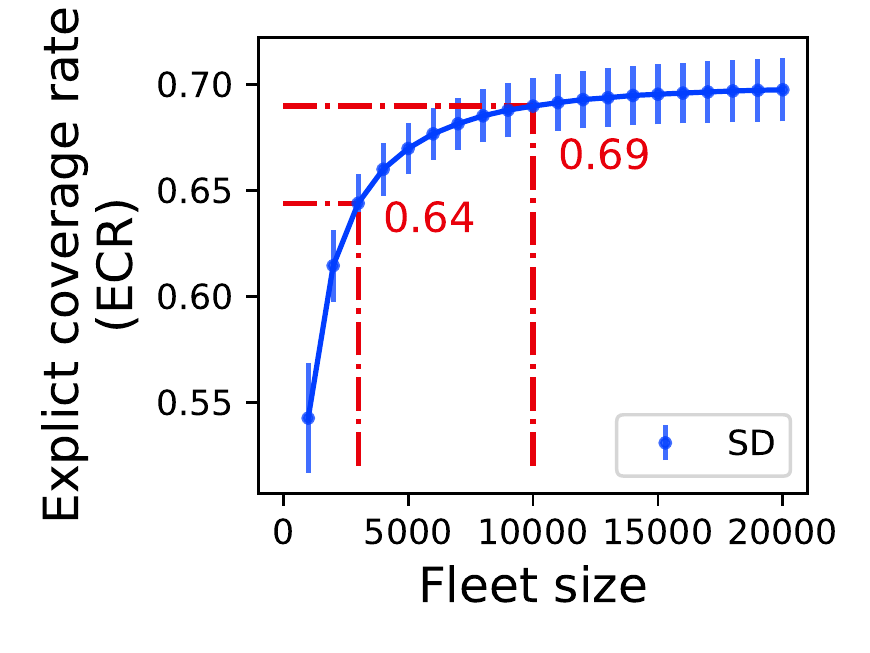}\label{fig:sensing_coverage}}
    \subfloat[Independence Score]{\includegraphics[width=0.16\textwidth]{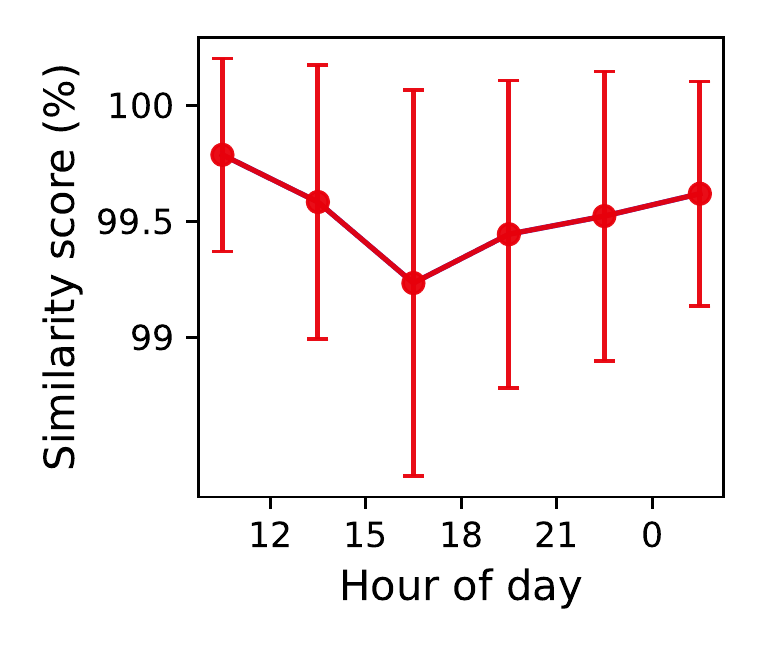}\label{fig:before_after_optimal_c}}
    \caption{Historical sensing performance (coverage, reliability, and independence)}\label{fig:before_after_inference}
\end{figure}

Fig.~\ref{fig:info_infer} shows the temporal distribution of coverage rates. We report that the morning periods (9 AM - 12 PM) have the lowest coverage rates due to commuters following fixed routes from residential areas to workplaces. However, the road segments in other areas are rarely visited, resulting in low ECRs. Moreover, the coverage on main roads is better than on branch roads during morning hours, especially in distant regions, which makes information inference challenging due to few visits and complex road networks. Similar trends are observed for ImpCRs, starting at the lowest level at 9 AM and plateauing from 6 PM to 9 PM. Over 90.8\% of links with inferring potential are already traversed, indicating high coverage on main roads.

Regarding sensing power across various sampled fleet sizes $K^{t}$ (ranging from 1,000 to 20,000 RVs), Fig.~\ref{fig:sensing_power} indicates a marginal effect in the sensing power, with 50\% of RVs achieving a sensing power of 0.90, close to the maximum. Also, a sensing power of 0.86 only requires 3,000 RVs (about 15\% of the fleet), indicating high reliability for daily operations. 
Fig.~\ref{fig:sensing_coverage} showcases the ECR, which closely approximates an exponential function. For instance, an ECR of 69.0\% can be achieved with a fleet of 10,000 RVs (representing half of the fleet), while a solid ECR of 64.0\% is sustained by a contingent of 3,000 RVs (equivalent to 15\% of the fleet). 

Further, we validate Assumption~\ref{assumption:independence} regarding the spatial independence of optimal trips by presenting the independence score $I^{t}$ in Fig.~\ref{fig:before_after_optimal_c}. Specifically, we adopt a rolling time window of 15 minutes within a three-hour period. Then, we calculate the independence score within the 15-minute rolling time window in different scenarios (e.g., time window threshold $\delta$ and the shortest-path search rage $K$). The average independence score exceeds 99.0\%, indicating a high level of independence in different scenarios and time periods. Furthermore, the standard deviations exhibited in the results underscore the robustness and reliability of the trip independence assumption.

\subsection{Results of optimal rerouting strategy}
In the following analyses, we compare the sensing performance improvement in terms of sensing coverage and reliability between the historical baseline and the optimal rerouting strategy. We examine three different time periods: morning, afternoon, and night. Specifically, we present the detailed sensing performance, travel mileage, and idling duration for different time periods, travel time threshold ($\delta$), and K-shortest path search range ($K$). The key metrics for sensing performance and related variables are as follows:
\begin{itemize}
    \item $r^{\text{exp}},r^{\text{inf}},r^{\text{imp}}$: ECR, InfCR, ImpCR
    \item $H^{t},S^{t}$: entropy and sensing power at time period $t$
    \item $B^{t}$: incentive for each RV driver to follow rerouting strategy
    \item $\Delta D, \Delta T$: changes in the travel distance/idle duration between the optimal rerouting and the historical trajectory (positive value indicates additional travel mileage).
\end{itemize}

\subsubsection{Comparison on the spatial distribution of sensing coverage}

\begin{figure}[h]
    \centering
    \subfloat[Spatial distribution of sensing coverage and frequency]{
    \includegraphics[width=0.45\textwidth]{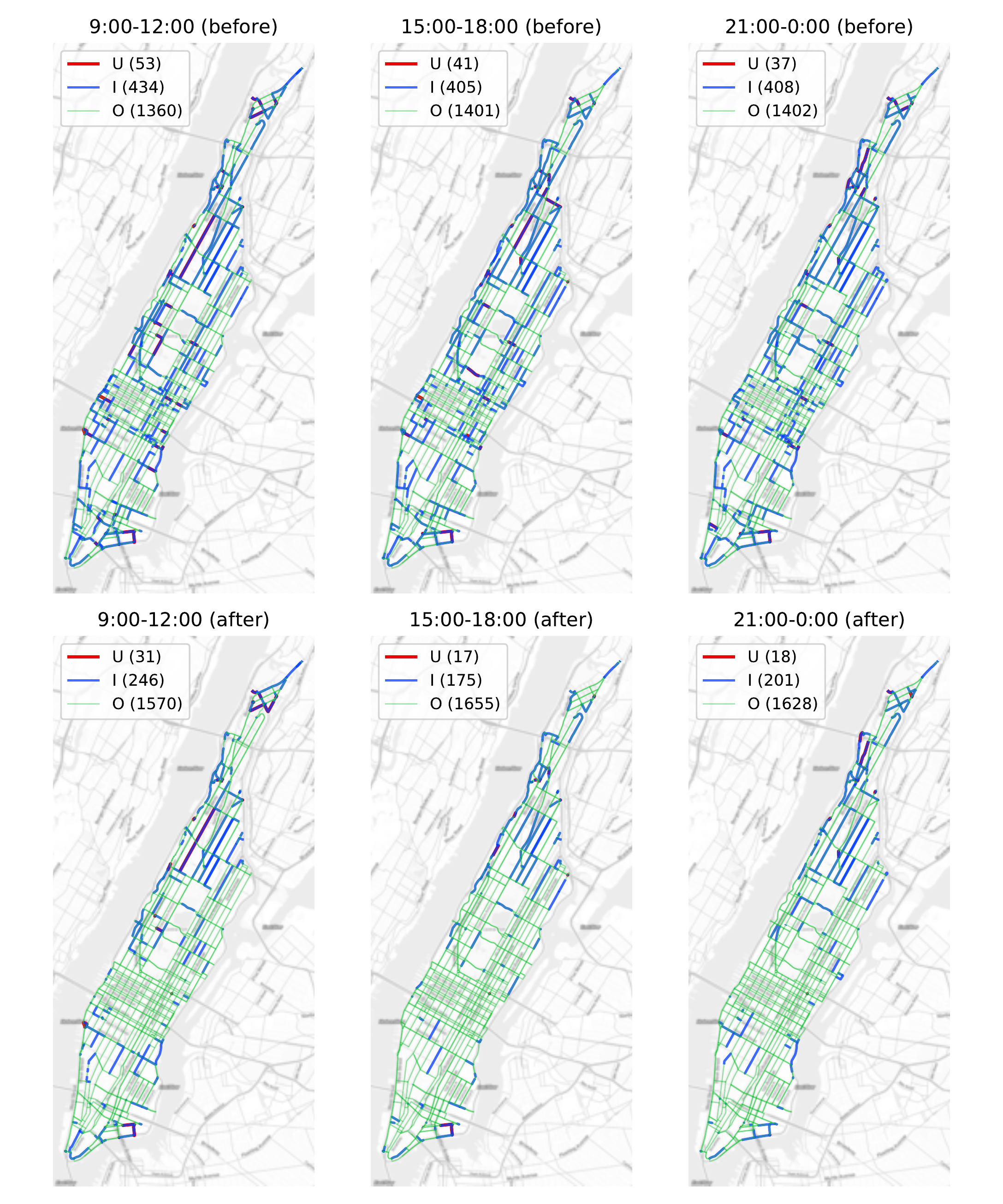}
    \label{fig:before_after_optimal_a}}\\
    \subfloat[Distribution of sensing frequency]{\includegraphics[width=0.45\textwidth]{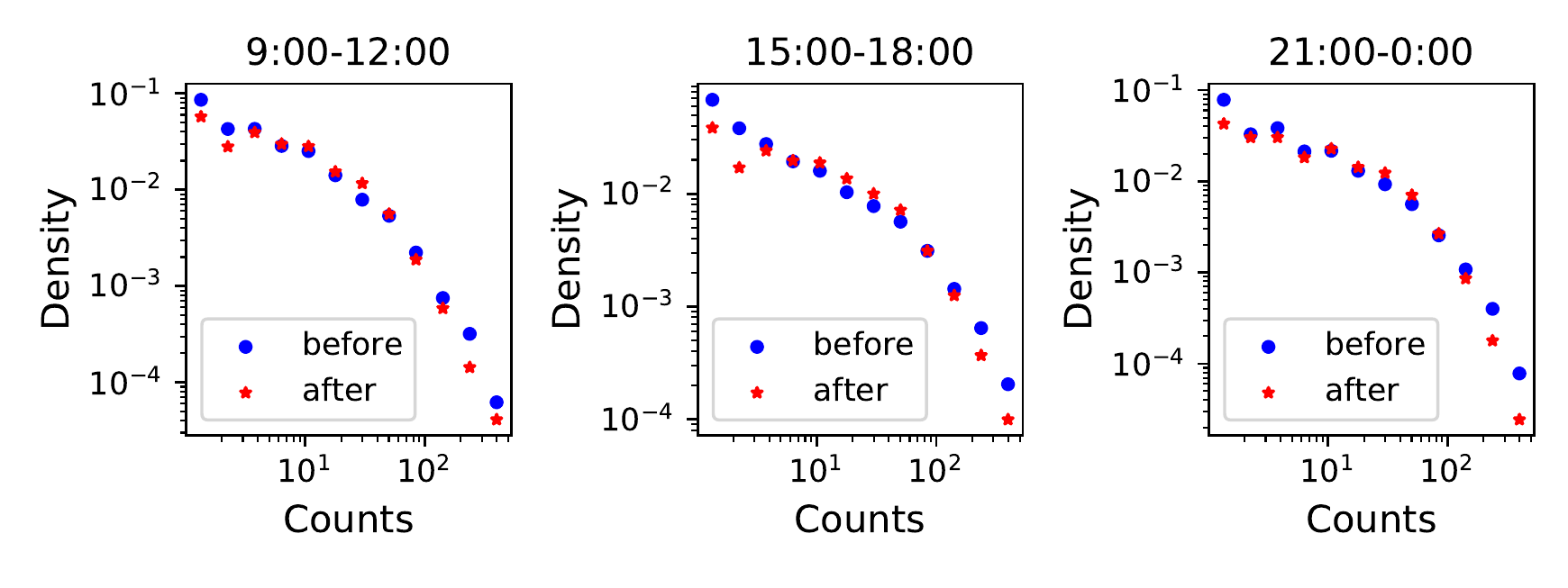}\label{fig:before_after_optimal_b}}
    \caption{Comparison between the historical trips and optimal rerouting trips ($\delta=1, K=20$, U: unknown, I: inferred, O: observed.)}\label{fig:before_after_optimal}
\end{figure}

We first define three types of links for better illustration: unknown (U), inferred (I), and observed (O) links to represent link observability. 
As shown in Fig.~\ref{fig:before_after_optimal}, after implementing the optimal rerouting strategy, a maximum of 254 additional links are directly visited, mainly in the Middle Manhattan area. This indicates that historical cruising trips in these areas likely followed repetitive routes, and rerouting encourages RVs to traverse unobserved and low-frequency links within their travel time threshold. In Upper Manhattan during the afternoon and evening periods (3-6 PM and 9 PM-12 AM), some links originally classified as unknown or inferred are now observed explicitly. However, the improvement in link observability is limited, particularly during the late-night period (9 PM-12 AM), due to few cruising trips in these areas.

Beyond link observability, we analyze the distribution of sensing frequency during the three time periods. Before rerouting, the distribution is highly skewed, with many links visited fewer than 25 times in the three-hour period. After rerouting, we observe a more even distribution, with less-visited links receiving more coverage, and a reduced density of links with low visit counts. Additionally, fewer links have high visit counts after rerouting, as RVs using repetitive paths now redistribute to cover additional links or links with low visit counts. This redistribution leads to improved explicit and implicit link coverage rates, as reflected in the distributions of unknown and inferred links before and after rerouting.

\subsubsection{Comparison on temporal distribution of sensing performance}

\begin{figure*}
\centering
    \includegraphics[width =\linewidth]{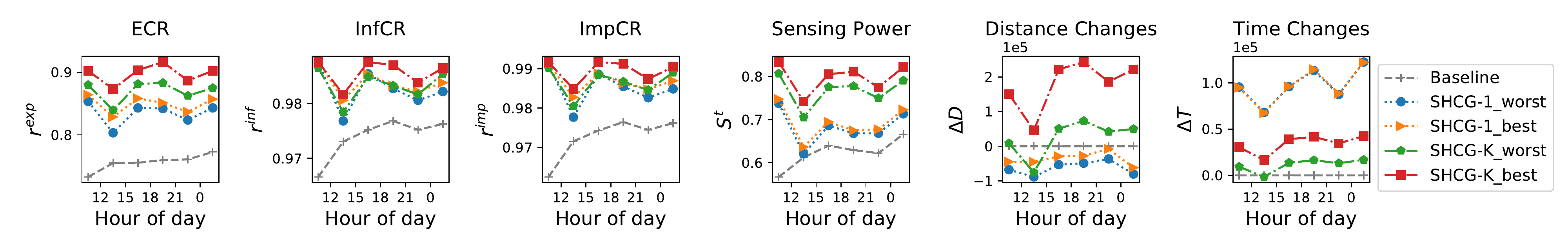}
    \caption{Comparison between SHCG-1, SHCG-K and historical baseline.}\label{fig:hist_comparison}
\end{figure*}

Fig.~\ref{fig:hist_comparison} compares the results of the optimal rerouting strategy and the historical baseline. To facilitate the discussion, we use "SHCG-K" to refer to the SHCG algorithm with multiple shortest path candidates ($K>1$), and "SHCG-1" for the case when $K=1$. The worst sensing performance in SHCG-K outperforms the best performance in SHCG-1 for all three time periods. For instance, during 9 PM - 12 AM, the best ECR and sensing power in SHCG-K under $\delta=1.2$ are 0.87 and 0.74, respectively, which are 3.6\% and 11.1\% higher than the most restrictive scenario in SHCG-1 ($\delta=0.9, K=1$).

Regarding travel distance and duration, all scenarios exhibit mileage savings in the rerouting trips and prolonged idle duration compared to the historical baseline and the results of SHCG-K and SHCG-1. Specifically, SHCG-1 tends to select the shortest path with the least sensing frequency, resulting in locally traversing road segments with fewer visits and more waiting time. In contrast, SHCG-K focuses globally on the entropy of sensing frequency, which may lead to deviations from the original path to achieve a higher entropy within the time window constraint.

\subsubsection{Sensitivity analyses on time window width $\delta$ and path search range $K$}

Table~\ref{tab:rerouting_performance_comparison} provides a comprehensive overview of the sensing performance outcomes resulting from our devised rerouting strategy. A larger $\delta$ indicates a wider time window from their original trajectories, potentially entailing longer travel times, while a larger value of $K$ expands the searching space of admissible alternative paths. Notably, instances of negative additional cruising distance values during morning periods reveal that our rerouting strategy can indeed yield shorter distances compared to historical trajectories, even in cases involving longer routes.

Focusing on scenarios with lower rerouting distances, we observe significant enhancements in ECR and sensing power, highlighting improvements of at least 14.5\% and 24.0\%, respectively. These promising outcomes extend beyond morning periods, holding true for both afternoon and night periods, thereby consistently showcasing the superiority of the optimal rerouting strategy. Remarkably, even within the most stringent conditions (i.e., $\delta=0.9$, $K=20$ in the afternoon time period), we report the enhancements in sensing coverage of up to 17.3\% in ECR $r^{\text{exp}}$, a minimum of 4.3\% increase in sensing entropy $H^{t}$, and an impressive 27.8\% improvement in sensing power $S^{t}$. 
\added{In comparison with the most stringent scenario, we observe a positive increase of 2.8\% in ECR and 3.8\% in sensing power, which comes with the trade-off of exhaustive computations for shortest paths across an expanded search space (e.g., higher values of $K$ and $\delta$). Hence, we conclude that a rerouting strategy with a high-quality solution remains attainable through a larger search space. Notably, even in the most stringent cases, remarkable enhancements compared to the historical baseline are achieved, highlighting the efficacy of our rerouting model.}

\begin{table}[h]
    \scriptsize
    \caption{Results of optimal rerouting strategy and historical baseline}
    \label{tab:rerouting_performance_comparison}
    \begin{tabular}{p{0.1cm}lp{0.19cm}p{0.19cm}p{0.19cm}p{0.19cm}p{0.19cm}rrp{0.19cm}p{0.19cm}p{0.19cm}}
    \toprule
    $t$& $(\delta,K)$ & $r^{\text{exp}}$ & $r^{\text{imp}}$ &  $r^{\text{inf}}$ &  $H^{t}$ &  $S^{t}$ &  $\Delta D$ &  $\Delta T$ &  \multicolumn{3}{c}{$B^{t}(\eta)$} \\\cmidrule(lr){10-12}
    &&&&&&&&& 0.2 & 0.3 & 0.4\\
    \midrule
   M &  Hist &   0.73 &   0.96 &     0.97 &     6.43 &          0.57 &               0.00 &                0.00 &       - &          -&         -  \\\hline
   M &  (0.9,20) &  0.84 &   0.98 &     0.98 &     6.70 &          0.71 &             -16.84 &               -1.18 &   0.1 &  0.19 &  0.39 \\
   M &  (0.9,40) &  0.84 &   0.98 &     0.98 &     6.72 &          0.71 &             -12.16 &                0.95 &  0.11 &  0.23 &  0.45 \\
   M &  (0.9,60) &  0.85 &   0.98 &     0.98 &     6.73 &          0.72 &              -9.74 &                2.05 &  0.12 &  0.25 &  0.48 \\
   M &  (0.9,80) &  0.85 &   0.98 &     0.98 &     6.74 &          0.72 &              -8.46 &                2.63 &  0.13 &  0.26 &   0.5 \\
   M &  (1.0,20) &  0.85 &   0.98 &     0.98 &     6.72 &          0.71 &             -11.18 &                1.39 &  0.11 &  0.23 &  0.44 \\
   M &  (1.0,40) &  0.85 &   0.98 &     0.98 &     6.74 &          0.72 &              -5.74 &                3.87 &  0.13 &  0.26 &  0.51 \\
   M &  (1.0,60) &  0.86 &   0.98 &     0.98 &     6.75 &          0.73 &              -2.77 &                5.21 &  0.15 &  0.29 &  0.56 \\
   M &  (1.0,80) &  0.86 &   0.98 &     0.98 &     6.76 &          0.73 &              -1.07 &                5.99 &  0.16 &  0.31 &   0.6 \\
   M &  (1.1,20) &  0.85 &   0.98 &     0.98 &     6.74 &          0.72 &              -6.65 &                3.45 &  0.17 &  0.46 &  1.48 \\
   M &  (1.1,40) &  0.86 &   0.98 &     0.98 &     6.76 &          0.73 &              -0.65 &                6.18 &  0.19 &  0.51 &  1.56 \\
   M &  (1.1,60) &  0.86 &   0.98 &     0.98 &     6.77 &          0.74 &               2.67 &                7.69 &  0.21 &  0.55 &  1.66 \\
   M &  (1.1,80) &  0.87 &   0.98 &     0.98 &     6.78 &          0.74 &               4.87 &                8.69 &  0.22 &  0.57 &  1.73 \\
   M &  (1.2,20) &  0.86 &   0.98 &     0.98 &     6.75 &          0.73 &              -2.51 &                5.33 &  0.18 &  0.49 &  1.52 \\
   M &  (1.2,40) &  0.86 &   0.98 &     0.98 &     6.78 &          0.73 &               4.18 &                8.37 &  0.21 &  0.54 &  1.64 \\
   M &  (1.2,60) &  0.87 &   0.98 &     0.98 &     6.78 &          0.74 &               8.16 &               10.18 &  0.22 &  0.57 &  1.69 \\
   M &  (1.2,80) &  0.87 &   0.98 &     0.98 &     6.80 &          0.74 &              10.14 &               11.08 &  0.24 &  0.61 &  1.78 \\\hline
   A &  Hist &    0.75 &   0.97 &     0.97 &     6.50 &          0.61 &               0.00 &                0.00 &       - &         - &         - \\\hline
   A &  (0.9,20) &  0.88 &   0.99 &     0.98 &     6.78 &          0.78 &              15.17 &                7.78 &  0.13 &  0.25 &  0.45 \\
   A &  (0.9,40) &  0.89 &   0.99 &     0.98 &     6.80 &          0.79 &              22.89 &               10.43 &  0.16 &   0.3 &  0.53 \\
   A &  (0.9,60) &  0.90 &   0.99 &     0.98 &     6.81 &          0.79 &              26.28 &               11.58 &  0.17 &  0.32 &  0.58 \\
   A &  (0.9,80) &  0.90 &   0.99 &     0.98 &     6.81 &          0.79 &              28.98 &               12.51 &  0.18 &  0.34 &  0.62 \\
   A &  (1.0,20) &  0.89 &   0.99 &     0.99 &     6.79 &          0.78 &              22.16 &               10.18 &  0.16 &  0.31 &  0.57 \\
   A &  (1.0,40) &  0.90 &   0.99 &     0.99 &     6.81 &          0.79 &              29.87 &               12.81 &  0.19 &  0.36 &  0.65 \\
   A &  (1.0,60) &  0.90 &   0.99 &     0.99 &     6.82 &          0.80 &              34.43 &               14.38 &   0.2 &   0.4 &  0.73 \\
   A &  (1.0,80) &  0.91 &   0.99 &     0.99 &     6.83 &          0.80 &              37.34 &               15.37 &  0.22 &  0.42 &  0.77 \\
   A &  (1.1,20) &  0.89 &   0.99 &     0.99 &     6.81 &          0.79 &              27.67 &               12.06 &  0.18 &  0.34 &  0.63 \\
   A &  (1.1,40) &  0.90 &   0.99 &     0.99 &     6.83 &          0.80 &              35.64 &               14.79 &  0.21 &   0.4 &  0.73 \\
   A &  (1.1,60) &  0.91 &   0.99 &     0.99 &     6.84 &          0.80 &              41.02 &               16.63 &  0.23 &  0.44 &   0.8 \\
   A &  (1.1,80) &  0.91 &   0.99 &     0.99 &     6.84 &          0.81 &              44.05 &               17.67 &  0.24 &  0.46 &  0.85 \\
   A &  (1.2,20) &  0.90 &   0.99 &     0.99 &     6.82 &          0.79 &              31.52 &               13.38 &  0.19 &  0.36 &  0.66 \\
   A &  (1.2,40) &  0.91 &   0.99 &     0.99 &     6.84 &          0.81 &              41.23 &               16.71 &  0.23 &  0.44 &   0.8 \\
   A &  (1.2,60) &  0.91 &   0.99 &     0.99 &     6.85 &          0.81 &              46.88 &               18.64 &  0.25 &  0.48 &  0.88 \\
   A &  (1.2,80) &  0.92 &   0.99 &     0.99 &     6.86 &          0.81 &              50.46 &               19.87 &  0.26 &  0.51 &  0.93 \\\hline
   N &  Hist &   0.76 &   0.97 &     0.98 &     6.49 &          0.64 &               0.00 &                0.00 &       - &        - &         - \\\hline
   N &  (0.9,20) &  0.88 &   0.99 &     0.99 &     6.82 &          0.79 &               9.94 &                7.42 &  0.09 &  0.17 &  0.32 \\
   N &  (0.9,40) &  0.88 &   0.99 &     0.98 &     6.84 &          0.80 &              15.77 &                9.35 &   0.1 &   0.2 &  0.38 \\
   N &  (0.9,60) &  0.88 &   0.99 &     0.99 &     6.85 &          0.80 &              17.97 &               10.07 &  0.11 &  0.21 &  0.39 \\
   N &  (0.9,80) &  0.88 &   0.99 &     0.99 &     6.86 &          0.80 &              19.70 &               10.65 &  0.11 &  0.22 &  0.41 \\
   N &  (1.0,20) &  0.88 &   0.99 &     0.99 &     6.84 &          0.79 &              16.33 &                9.53 &   0.1 &  0.19 &  0.35 \\
   N &  (1.0,40) &  0.89 &   0.99 &     0.98 &     6.86 &          0.80 &              23.35 &               11.85 &  0.12 &  0.22 &  0.41 \\
   N &  (1.0,60) &  0.89 &   0.99 &     0.99 &     6.87 &          0.81 &              26.49 &               12.89 &  0.12 &  0.24 &  0.44 \\
   N &  (1.0,80) &  0.89 &   0.99 &     0.99 &     6.87 &          0.81 &              27.76 &               13.31 &  0.13 &  0.25 &  0.46 \\
   N &  (1.1,20) &  0.89 &   0.99 &     0.99 &     6.86 &          0.80 &              22.07 &               11.43 &  0.11 &  0.21 &  0.38 \\
   N &  (1.1,40) &  0.89 &   0.99 &     0.99 &     6.88 &          0.81 &              31.09 &               14.41 &  0.13 &  0.26 &  0.48 \\
   N &  (1.1,60) &  0.89 &   0.99 &     0.99 &     6.89 &          0.81 &              34.40 &               15.50 &  0.14 &  0.28 &  0.52 \\
   N &  (1.1,80) &  0.90 &   0.99 &     0.99 &     6.90 &          0.82 &              37.38 &               16.48 &  0.15 &  0.29 &  0.55 \\
   N &  (1.2,20) &  0.89 &   0.99 &     0.99 &     6.88 &          0.81 &              27.99 &               13.38 &  0.12 &  0.22 &   0.4 \\
   N &  (1.2,40) &  0.89 &   0.99 &     0.99 &     6.90 &          0.82 &              37.71 &               16.59 &  0.15 &  0.28 &  0.52 \\
   N &  (1.2,60) &  0.90 &   0.99 &     0.99 &     6.91 &          0.82 &              41.22 &               17.75 &  0.16 &   0.3 &  0.57 \\
   N &  (1.2,80) &  0.90 &   0.99 &     0.99 &     6.92 &          0.82 &              44.36 &               18.79 &  0.17 &  0.32 &  0.61 \\
\bottomrule
\end{tabular}
\footnotesize{\\\\Note: M: Morning, 9 AM - 12 PM; A: Afternoon: 3 - 6 PM; N: Night, 9 PM - 12 AM; Hist: historical baseline from the trip record data.}
\end{table}

Moreover, our rerouting approach's effectiveness is reinforced by its ability to cover as much as 92\% of road segments and achieve the InfCR of 99\%. This highlights the strategic value of our approach in guiding RV trajectories, enabling focused exploration of less-traveled links and implicit inference. Importantly, these merits are achieved while optimizing RV deployment within current capacity constraints, ensuring a balanced interplay between improved sensing performance and the preservation of service quality.

\subsubsection{Sensitivity analyses on incentive}\label{sec:sen_on_budget}
We now delve into the sensitivity analyses of our incentive mechanism, represented by the scaling parameter $\eta$. Fig.~\ref{fig:trade_off_budget} illustrates the interplay between the enhancement in ECR, denoted as $\Delta r^{\text{exp}}$, and the corresponding increase in the allocated incentive $B^{t}$. Generally, when employing a broader search range $K$ or a less restrictive time threshold $\delta$, we observe a consistent ECR improvement of no less than 1.5\%. Meanwhile, this improvement comes at the cost of an additional incentive cost of only up to \$0.13 ($\eta=0.2$). Furthermore, the impact of enlarging $K$ on ECR enhancement becomes marginal, as depicted in Fig.~\ref{fig:trade_off_K}. This suggests that searching up to $K=60$ shortest paths adequately covers the optimal routes, while further expansion encounters travel time constraints. Nevertheless, the incentive cost augmentation demonstrates a linear relationship with $K$ under constant $\delta=1.0$, especially for $\eta=0.2$ and $0.3$, across all time periods (excluding $\eta=0.4$ during 9 PM - 12 AM). This signifies the identification of extended detour routes ($\Delta D_{k}$), necessitating exponentially increasing incentive costs to achieve heightened sensing coverage.
\begin{figure}[H]
    \subfloat[Sensitivity analyses on $K$ under the fixed $\delta=1.0$]{\includegraphics[width=1.0\linewidth]{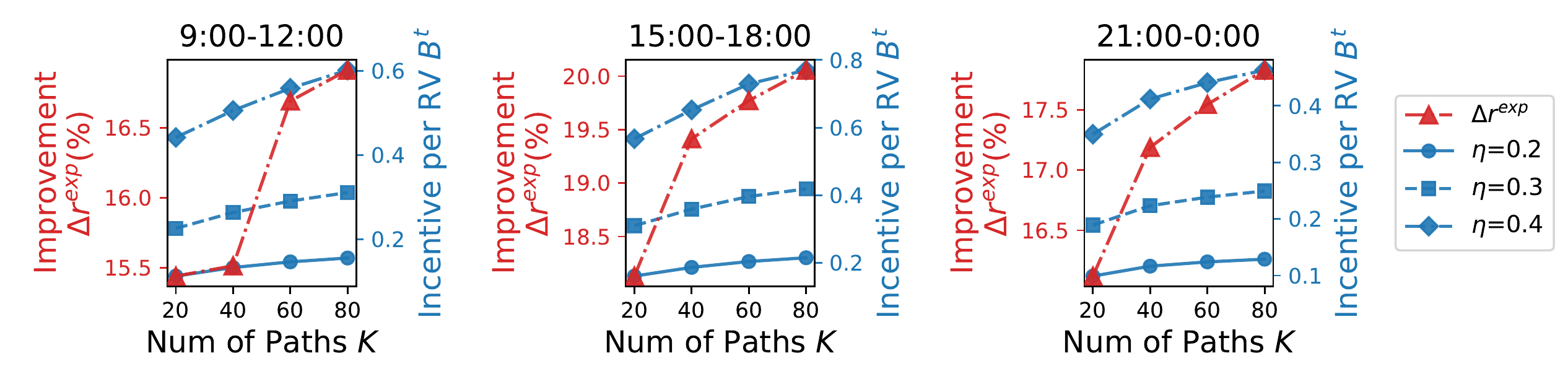}\label{fig:trade_off_K}}\\
    \subfloat[Sensitivity analyses on $\delta$ under the fixed $K=60$]{\includegraphics[width=1.0\linewidth]{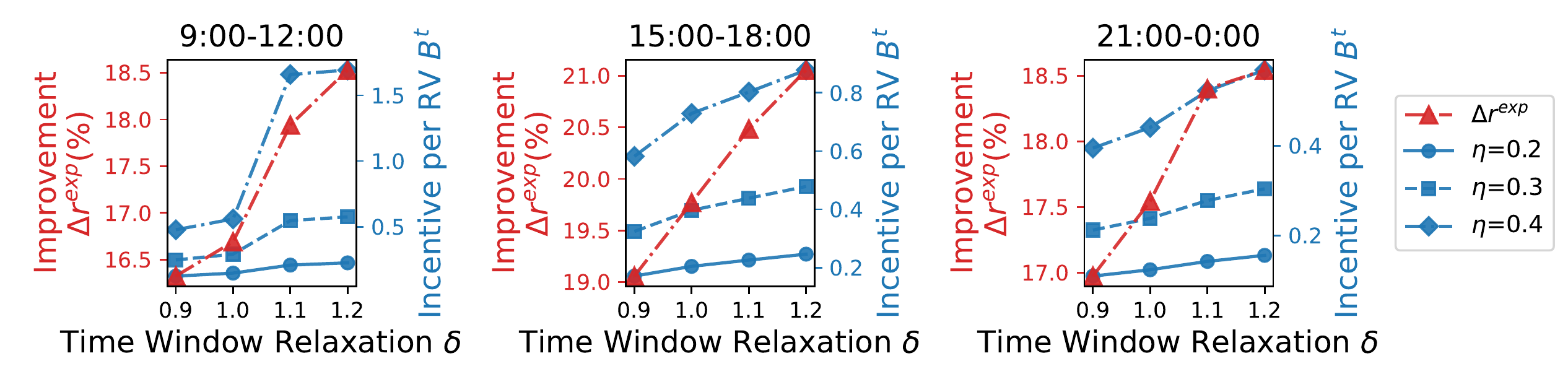}\label{fig:trade_off_delta}}
    \caption{Trade off between improvement of sensing coverage $\Delta r^{\text{exp}}$ (\%) and incentive for each RV trip $B^{t}$ (unit: US dollar)}
    \label{fig:trade_off_budget}
\end{figure}
In Fig.~\ref{fig:trade_off_delta}, the most substantial improvement in ECR (up to 4.5\%) in the afternoon requires a modest additional incentive cost as low as \$0.13 at $\eta=0.2$. Similarly, the morning and night periods exhibit ECR enhancements of 3.6\% and 2.3\%, respectively, with at least \$0.14 and \$0.08 incentive cost increments under $\eta=0.2$. These findings underscore the remarkable potential of our rerouting strategy. Particularly during mornings, given the robust sensing coverage of main roads (as evident in Fig.~\ref{fig:before_after_optimal_a}), a slight deviation from the original route suffices to traverse necessary road segments (e.g., nearby branch roads). Remarkably, in the night period, our optimal rerouting approach effectively covers over 88\% of road segments under $\delta=1.0$ and $K=20$, with only an incentive cost as low as \$ 0.1 per RV driver. This substantiates the practical viability of our rerouting solution for effective sensing tasks.

\section{Conclusion}\label{sec:conclusions}
This study introduces a novel approach for deploying large-scale RV fleets to facilitate drive-by sensing within urban road networks. Initial exploration focuses on historical sensing performance, examining coverage and reliability. Subsequently, an optimal rerouting strategy is proposed to guide unoccupied RVs along specific unvisited links, enhancing both coverage and reliability. The strategy is formulated as a trip-based RVRP, with the objective of maximizing sensing entropy to enhance the sensing coverage and sensing power, thereby reinforcing the overall sensing performance. Further, an SHCG approach is designed to generate high-quality rerouting paths for unoccupied RVs in small time intervals.

Results demonstrate the superiority of our optimal rerouting strategy over the historical baseline, enhancing ECR by a minimum of 15.1\% and sensing power by up to 27.9\%. Sensitivity analyses unveil the positive impact of relaxed time constraints and broader shortest path search ranges, boosting coverage by up to 3.6\%. Importantly, even in the most stringent scenarios in the SHCG algorithm, significant enhancements relative to the historical baseline are attained, thereby underscoring the effectiveness of our rerouting model. In addition, only a \$0.13 incentive cost per RV driver yields a 4.5\% increase in ECR and 3.8\% enhancement in sensing power, highlighting a cost-effective route to heightened sensing performance.

Finally, we note that the fleet deployment for sensing may individually sacrifice the driver income but will achieve an overall high performance on the link coverage. The results imply a promising direction for the coordinated operation of connected and autonomous RVs considering network sensing as a subtask of daily operations.

\bibliographystyle{IEEEtranN}
\bibliography{ref}

\begin{IEEEbiography}[
{\includegraphics[width=1in,height=1.25in,clip,keepaspectratio]{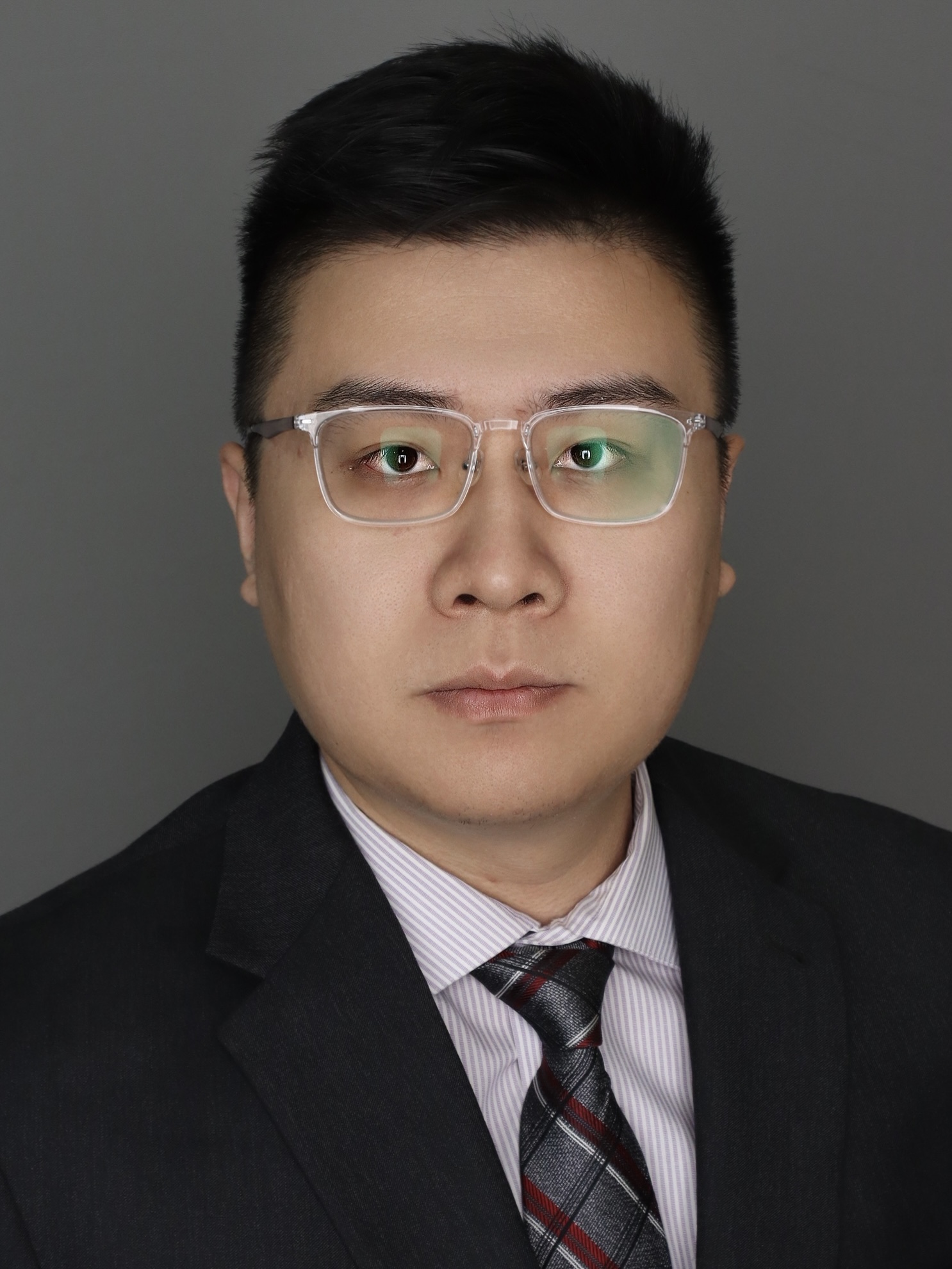}}
]
{Shuocheng Guo} received the B.E degree in civil engineering from Central South University, Changsha, China, and the M.S. degree in transportation engineering from University of Illinois at Urbana and Champaign, Urbana, IL, USA. He is currently working towards the Ph.D. degree in transportation engineering at The University of Alabama, Tuscaloosa, AL, USA. His research interests include combinatorial optimization, electrified transportation network, and urban computing.
\end{IEEEbiography}
\begin{IEEEbiography}
[
{\includegraphics[width=1in,height=1.25in,clip,keepaspectratio]{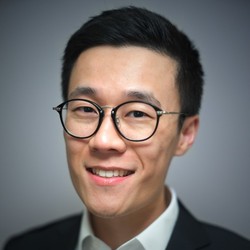}}
]
{Xinwu Qian} received the B.S. degree in transportation engineering from Tongji University, Shanghai,
China, and the M.S. and Ph.D. degrees in transportation engineering from Purdue University, West
Lafayette, IN, USA. He is currently an Assistant
Professor of Civil, Construction, and Environmental
Engineering with The University of Alabama. His
research interests include big data analytics, complex network analysis, and network modeling and optimization.
\end{IEEEbiography}

\end{document}